\documentclass[10pt, twocolumn]{article}

\usepackage[authoryear,sort&compress,longnamesfirst,round]{natbib}
\input{arxiv_load_packages}
\newcommand{\E}{\mathbb{E}}

\newcommand{\R}{\mathbb{R}}

\newcommand{\T}{\mathbb{T}}

\newcommand{\Prob}{\mathbb{P}} 
\renewcommand{\Prob}[1]{\mathbb{P}\left( #1 \right)} 

\newcommand{\ProbState}[2]{\mathbb{P}_{#1}\left( #2 \right)}

\DeclareMathOperator*{\argmax}{arg\,max} 
\DeclareMathOperator*{\argmin}{arg\,min} 

\newcommand{\ProbHat}[1]{\widehat{\mathbb{P}}\left( #1 \right)} 

\newcommand{\Indicator}[1]{\mathbf{1}\left( #1 \right)}

\def\independenT#1#2{\mathrel{\setbox0\hbox{$#1#2$}%
\copy0\kern-\wd0\mkern4mu\box0}}

\newcommand{\method}{LICORS }

\newcommand{\vnorm}[1]{\lVert #1 \rVert}

\newcommand{\card}[1]{| #1 |}

\newcommand{\KL}[2]{\mathcal{D}_{KL}\left(#1 \mid \mid #2 \right)} 
\newcommand{\SymKLexp}[2]{\frac{\KL{p}{q} + \KL{q}{p}}{2}}

\newcommand{\vecS}[1]{\mathbf{#1}}
\newcommand{\field}[3]{#1 \hspace{-0.075cm} \left(\vecS{#2}, #3 \right) }

\newfloat{myalgo}{tbhp}{mya}

{\begin{myalgo}[#1]
\centering
\begin{minipage}{0.9\textwidth}
\begin{algorithm}[H]}%
{\end{algorithm}
\end{minipage}

\end{myalgo}}

\newcounter{bar}
\newcommand{\foo}{%
\stepcounter{bar}%
\thebar}

\newtheorem{theorem}{Theorem}[section]

\newtheorem{assumption}[theorem]{Assumption}

\newtheorem{definition}[theorem]{Definition}

\newtheorem{proposition}[theorem]{Proposition}

\graphicspath{{images/}}
\DeclareGraphicsExtensions{.pdf}

\usepackage[margin = 2.3cm, top = 2cm]{geometry}

\newcommand{\Title}{Mixed LICORS: A Nonparametric Algorithm for Predictive State Reconstruction\footnote{Appears in the Proceedings of AISTATS 2013.}}
\newcommand{\Author}{Georg M.\ Goerg and Cosma Rohilla Shalizi}
\newcommand{\Affil}{ Department of Statistics, Carnegie Mellon University, Pittsburgh, PA 15213\\ \texttt{ $\lbrace$ gmg, cshalizi $\rbrace$ @stat.cmu.edu}}

\title{\Title}
\author{\Author \\ \Affil}
\date{\today}

\begin{document}

\maketitle

\begin{abstract}
  We introduce \emph{mixed LICORS}, an algorithm for learning nonlinear,
  high-dimensional dynamics from spatio-temporal data, suitable for both
  prediction and simulation.  Mixed LICORS extends the recent LICORS algorithm
  \citep{LICORS} from hard clustering of predictive distributions to a
  non-parametric, EM-like soft clustering.  This retains the asymptotic
  predictive optimality of LICORS, but, as we show in simulations, greatly
  improves out-of-sample forecasts with limited data.  The new method is
  implemented in the publicly-available R package
  \href{http://cran.r-project.org/web/packages/LICORS/}{LICORS}.
\end{abstract}

\section{Introduction}
\label{sec:Introduction}

Recently \citet{LICORS} introduced light cone reconstruction of states
(LICORS), a nonparametric procedure for recovering predictive states from
spatio-temporal data.  Every spatio-temporal process has an associated, latent
prediction process, whose measure-valued states are the optimal local
predictions of the future at each point in space and time \citep{Shalizi03}.
LICORS consistently estimates this prediction process from a single realization
of the manifest space-time process, through agglomerative clustering in the
space of predictive distributions; estimated states are clusters.  This
converges on the minimal set of states capable of optimal prediction of the
original process.

Experience with other clustering problems shows that soft-threshold techniques
often predict much better than hard-threshold methods.  Famously, while
$k$-means \citep{Lloyd82_KmeansLloydAlgorithm} is very fast and robust, the
expectation maximization (EM) algorithm
\citep{Dempster-Laird-Rubin-introduce-EM} in Gaussian mixture models gives
better clustering results.  Moreover, mixture models allow clusters to be
interpreted probabilistically, and new data to be assigned to old clusters.

With this inspiration, we introduce {\em mixed LICORS}, a soft-thresholding
version of (hard) LICORS.  This embeds the prediction process framework of
\citet{Shalizi03} in a mixture-model setting, where the predictive states
correspond to the optimal mixing weights on extremal distributions, which are
themselves optimized for forecasting.  Our proposed nonparametric, EM-like
algorithm then follows naturally.

After introducing the prediction problem and fixing notation, we explain the
mixture-model and hidden-state-space interpretations of predictive states (\S
\ref{sec:light_cones_predictive_states}).  We then present our nonparametric EM
algorithm for estimation, with automatic selection of the number of predictive
states (\S \ref{sec:EM}), and the corresponding prediction algorithm (\S
\ref{sec:prediction_mixedLICORS}).  After demonstrating that mixed LICORS
predicts better out of sample than hard-clustering procedures (\S
\ref{sec:simulations}), we review the proposed method and discuss future work
(\S \ref{sec:discussion}).

\section{A Predictive States Model for Spatio-temporal Processes}
\label{sec:light_cones_predictive_states}

We fix notation and the general set-up for predicting spatio-temporal
processes, following \citet{Shalizi03}.  We observe a random field
$\field{X}{r}{t}$, discrete- or continuous-valued, at each point $\vecS{r}$ on
a regular spatial lattice $\vecS{S}$, at each moment $t$ of discrete time $\T =
1:T$, or $N=T\card{\vecS{S}}$ observational in all.  The field is $(d+1)D$ if
space $\vecS{S}$ is $d$ dimensional (plus $1$ for time); video is $(2+1)D$.
$\vnorm{ \vecS{r} -\vecS{u} }$ is a norm (e.g., Euclidean) on the spatial
coordinates $\vecS{r}, \vecS{u} \in \vecS{S}$.

To optimally predict an unknown (future) $\field{X}{r}{t}$ given (past)
observed data $\mathcal{D} = \lbrace \field{X}{s}{u} \rbrace_{\vecS{s} \in
  \vecS{S}, u \in \T}$, we need to know $\Prob{ \field{X}{r}{t} \mid
  \mathcal{D} }$.  Estimating this conditional distribution is quite difficult
if $\field{X}{r}{t}$ can depend on all variables in $\mathcal{D}$.  Since
information in a physical system can typically only propagate at a finite speed
$c$, $\field{X}{r}{t}$ can only depend on a subset of $\mathcal{D}$. We
therefore only have to consider local, spatio-temporal neighborhoods of
$\field{}{r}{t}$ as possibly relevant for prediction.

\subsection{Light Cones}

The past light cone (PLC) of $\field{}{r}{t}$ is all past events that could
have influenced $\field{}{r}{t}$,
\begin{equation}
\label{eq:PLC}
\field{L^{-}}{r}{t} = \left\{\field{X}{u}{s} \mid s \leq t, \vnorm{ 	\vecS{r}-\vecS{u} } \leq c (t-s)\right\}.
\end{equation}
Analogously the future light cone (FLC) $\field{L^{+}}{r}{t}$ is all future
events which could be influenced by $\field{}{r}{t}$.  In practice, we limit
PLCs and FLCs to finite horizons $h_p$ and $h_f$; together with $c$, these fix
the dimensionality of light-cones; see Figure \ref{fig:light_cones_2D_and_3D}
for an illustration.

\begin{figure}[!t]
  \centering
  \begin{subfigure}[t]{0.17\textwidth}
    \centering
    \includegraphics[width=\textwidth, trim = 0.8cm 0.7cm 0.8cm 0.1cm, clip = true]{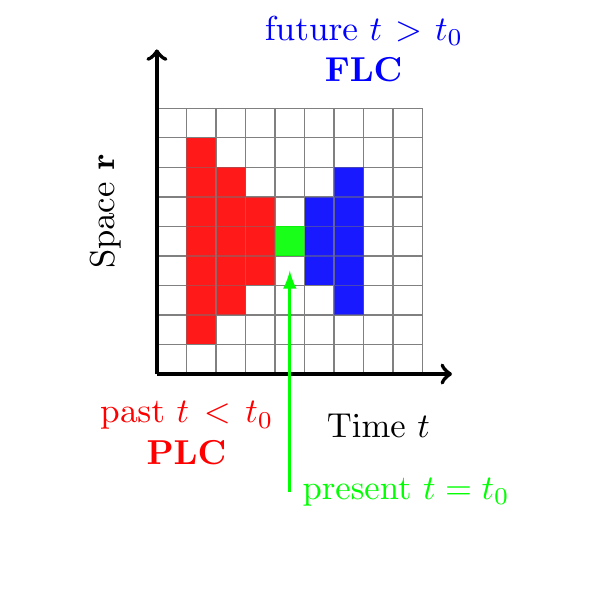}\label{fig:light_cones_1plus1D}
    \caption{$(1+1)D$ field}
  \end{subfigure}%
  ~ 
  \begin{subfigure}[t]{0.3\textwidth}
    \centering
    \includegraphics[width=\textwidth, trim = 0.6cm 0.8cm 0.2cm 0.1cm, clip = true]{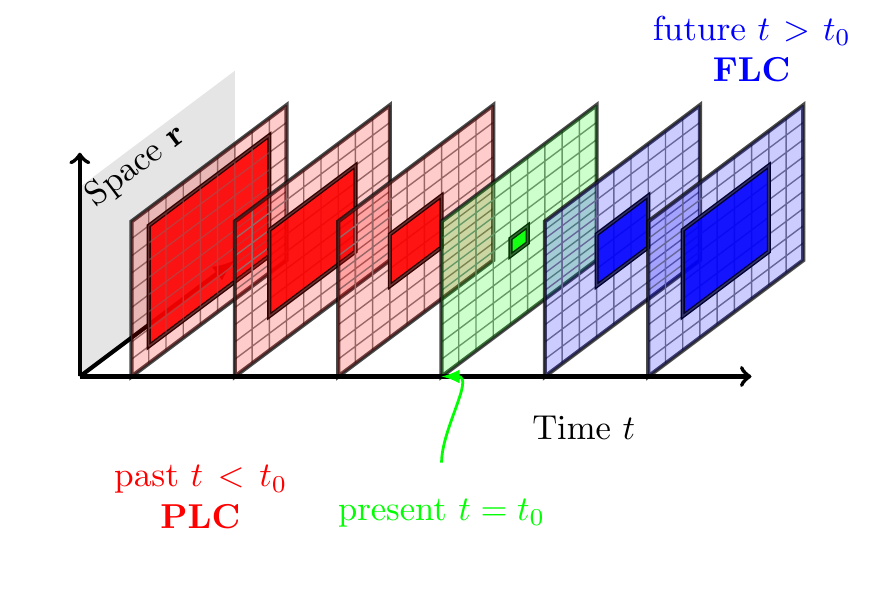}\label{fig:light_cones_2plus1D}
    \caption{$(2+1)D$ field}
  \end{subfigure}
  \caption{\label{fig:light_cones_2D_and_3D} Geometry of past (
    {\color{red}{$\field{\ell^{-}}{r}{t}$}}) and future
    ({\color{blue}{$\field{\ell^{+}}{r}{t}$}}) light cones, with $h_p = 3$,
    $h_f = 2$, $c=1$.}
\end{figure}

For physically-plausible processes, then, we only have to find the conditional
distribution of $\field{X}{r}{t}$ given its PLC $\left( \ell^{-},
  \field{}{r}{t} \right)$. Doing this for every $\field{}{r}{t}$ in space-time,
however, is still excessive.  Not only do spatio-temporal patterns recur, but
different histories can have similar predictive consequences.  Thus we need not
find $\Prob{ \field{X}{r}{t} \mid \left( \ell^{-}, \field{}{r}{t} \right)}$ for
each $\field{}{r}{t}$ separately, but can first summarize PLCs by predictively
sufficient statistics, $\epsilon\left( \ell^{-}, \field{}{r}{t} \right)$, and
then only find $\Prob{ \field{X}{r}{t} \mid \epsilon\left( \ell^{-},
    \field{}{r}{t} \right)}$.

We now construct the minimal sufficient statistic,
$\epsilon(\field{\ell^{-}}{r}{t})$, following \citet{Shalizi03}, to which we
refer for some mathematical details, and references on predictive sufficiency.

\begin{definition}[Equivalent configurations]
  \label{def:equivalent_configurations}
  The past configurations $\ell_i^{-}$ at $\field{}{r}{t}$ and $\ell_j^{-}$ at
  $\field{}{u}{s}$ are {\em predictively equivalent}, $(\ell_i^{-},
  \field{}{r}{t}) \sim (\ell_j^{-}, \field{}{u}{s})$, if they have the same
  conditional distributions for FLCs, i.e.,
  \small
  \begin{equation*}
    \Prob{\field{L^{+}}{r}{t} \mid \field{L^{-}}{r}{t} = \ell_i^{-} } = \Prob{ \field{L^{+}}{u}{s}  \mid \field{L^{-}}{u}{s} = \ell_j^{-} }.
  \end{equation*}
  \normalsize
\end{definition}

Let $\left[ (\ell^{-}, \field{}{r}{t}) \right]$ be the equivalence class of
$(\ell^{-}, \field{}{r}{t})$, the set of all past configurations and
coordinates that predict the same future as $\ell^{-}$ does at
$\field{}{r}{t}$.  Let
\begin{equation}
  \label{eq:epsilon}
  \epsilon(\ell^{-},\field{}{r}{t}) \equiv \left[ \ell^{-} \right] .
\end{equation} 
be the function mapping each $(\ell^{-}, \field{}{r}{t})$ to its predictive
equivalence class.  The values $\epsilon$ can take are {\em predictive states};
they are the minimal statistics which are sufficient for predicting $L^+$ from
$L^-$ \citep{Shalizi03}.

\begin{assumption}[Conditional invariance]
\label{ass:cond_invariance}
  The predictive distribution of a PLC configuration $\ell^{-}$ does not change
  over time or space.  That is, for all $\vecS{r}, t$, all $\vecS{u}, s$, and
  all past light-cone configurations $\ell^{-}$,
  \begin{equation}
    (\ell^{-},\field{}{r}{t}) \sim (\ell^{-},\field{}{u}{s})
  \end{equation}
  \label{ass:st-invariance}
  We may thus regard $\sim$ as an equivalence relation among PLC
  configurations, and $\epsilon$ as a function over $\ell^{-}$ alone.
\end{assumption}

Assumption \ref{ass:cond_invariance} enables inference from a single
realization of the process (as opposed to multiple independent
realizations). For readability, we encode the space-time index $\field{}{r}{t}$
by a single index $i = 1, \ldots, N$.\footnote{Appendix \ref{sec:proofs} in the
  SM gives an explicit rule (Eq.\ \eqref{eq:assign_index_increasingy}) to map
  each $\field{}{r}{t}$ to a unique $i$.}

The future is independent of the past given the predictive state
\citep{Shalizi03},
\begin{equation}
\label{eq:conditional_independence}
\Prob{ L^{+}_i \mid \ell^{-}_i, \epsilon(\ell^{-}_i) } = \Prob{L^{+}_i \mid \epsilon(\ell^{-}_i) }
\end{equation} 
and, by construction,
\begin{equation}
\label{eq:equivalent_predictions}
\Prob{ L^{+}_i \mid \ell^{-}_i } = \Prob{L^{+}_i \mid \epsilon(\ell^{-}_i) }.
\end{equation}

Thus the prediction problem simplifies to estimating $\Prob{ X_i \mid
  \epsilon(\ell^{-}_i) }$ for each $i$.

\subsection{A Statistical Predictive-States Model}

The focus on predictive distributions in constructing the predictive states
does not prevent us from using them to model the joint distribution.

\begin{proposition}
  \label{prop:joint_factorizes}
  The joint pdf of $X_1, \ldots, X_{\tilde{N}}$ satisfies
  \begin{equation}
    \label{eq:joint_factorizes}
    \Prob{X_1, \ldots, X_{\tilde{N} } } = \Prob{ \vecS{M} } \prod_{i=1}^{\tilde{N}} \Prob{X_i \mid \ell^{-}_i},
  \end{equation}
  where $\vecS{M}$ is the {\bf margin} of the spatio-temporal process, i.e.,
  all points in the field that do not have a completely observed PLC.
\end{proposition}
\begin{proof}
  See Supplementary Material, \S \ref{proof:prop:joint_factorizes}.
\end{proof}

Proposition \ref{prop:joint_factorizes} shows that the focus on conditional
predictive distributions does not restrict the applicability of the light cone
setup to forecasts alone, but is in fact a generative representation for any
spatio-temporal process.

\subsubsection{Minimal Sufficient Statistical Model}

Given $\epsilon(\ell^{-}_i)$ Eq.\ \eqref{eq:joint_factorizes} simplifies to
(omitting $\Prob{\vecS{M}}$)
\begin{align}
  \Prob{X_1, \ldots, X_N \mid \vecS{M}, \epsilon} & = \prod_{i=1}^{N} \Prob{X_{i} \mid  \ell_i^{-}; \epsilon(\ell^{-}_i) } \nonumber\\
  \label{eq:joint_pdf_product_of_epsilons}
  &= \prod_{i=1}^{N} \Prob{X_{i} \mid \epsilon(\ell^{-}_i) },
\end{align}
using \eqref{eq:conditional_independence}.  Any particular $\epsilon$
implicitly specifies the number of predictive states $K$, and all $K$
predictive distributions $\Prob{X_i \mid \epsilon(\ell^{-}_i)}$. However, in
practice only $X_i$ and $\ell^{-}_i$ are observed; the mapping $\epsilon$ is
exactly what we are trying to estimate.

\subsection{Predictive States as Hidden Variables}

Since there is a one-to-one correspondence between the mapping $\epsilon$ and
the set of equivalence classes / predictive states $\epsilon(\ell^{-})$ which
are its range, \citet{Shalizi03} and \citet{LICORS} do not formally distinguish
between them.  Here, however, it is important to keep distinction between the
predictive state space $\mathcal{S}$ and the mapping $\epsilon: \ell^{-}_i
\rightarrow \mathcal{S}$.  Our EM algorithm is based on the idea that the
predictive state is a hidden variable, $S_i$, taking values in the finite state
space $\mathcal{S} = \lbrace s_1, \ldots, s_K \rbrace$, and the mixture weights
of PLC $\ell^{-}_i$ are the soft-threshold version of the mapping
$\epsilon(\ell^{-}_i)$.  It is this hidden variable interpretation of
predictive states that we use to estimate the minimal sufficient statistic
$\epsilon$, the hidden state space $\mathcal{S}$, and the predictive
distributions $\Prob{ X_i \mid S_i = s_j}$.

Introducing the abbreviation $\ProbState{j}{\cdot}$ for $\Prob{ \cdot \mid S_i
  = s_j}$, the latent variable approach lets
\eqref{eq:joint_pdf_product_of_epsilons} be written as
\begin{equation}
\label{eq:joint_pdf_product_of_pred_states} 
\prod_{i=1}^{N} \Prob{X_{i} \mid  S_i } = \prod_{i=1}^{N}{ \sum_{j=1}^{K}{\Indicator{S_i = s_j } \ProbState{j}{X_{i}} }}.
\end{equation}
Eq.\ \eqref{eq:joint_pdf_product_of_pred_states} is the pdf of a $K$ component
mixture model with \emph{complete data}, and $\Indicator{S_i = s_j }$ is a
randomized version of $\epsilon: \ell^{-}_i \mapsto S_i$.

\subsection{Log-likelihood of $\epsilon$}

From \eqref{eq:joint_pdf_product_of_pred_states} the complete data
log-likelihood is, neglecting a $\log{\Prob{\mathbf{M}}}$ term,
\small
\begin{align}
  \ell(\epsilon; \mathcal{D}, S_1^N) &= \sum_{i=1}^{N} \log \left( \sum_{j=1}^{K} \Indicator{S_i = s_j } \Prob{X_{i} \mid  \epsilon(\ell^{-}_i) = s_j } \right) \nonumber \\
  \label{eq:predictive_loglik_complete}
  & = \sum_{i=1}^{N} \sum_{j=1}^{K} \Indicator{S_i = s_j} \log \Prob{X_{i} \mid
    \epsilon(\ell^{-}_i) = s_j},
\end{align}
\normalsize
where $S_1^N := \lbrace S_1, \ldots, S_N \rbrace$ and the second equality
follows since $\Indicator{S_i = s_j} = 1$ for one and only one $j$, and $0$
otherwise.

The ``parameters'' in \eqref{eq:predictive_loglik_complete} are $\epsilon$ and
$K$; $X_i$ and $\ell^{-}_i$ are observed, and $S_i$ is a hidden variable.  The
optimal mapping $\epsilon: L^{-} \rightarrow \mathcal{S}$ is the one that
maximizes \eqref{eq:predictive_loglik_complete}:
\begin{equation}
\label{eq:optim_problem_epsilon}
\epsilon^{*} = \argmax_{\epsilon}{\ell(\epsilon; \mathcal{D}, S_i)}.
\end{equation}

Without any constraints on $K$ or $\epsilon$ the maximum is obtained for $K =
N$ and $\epsilon(\ell^{-}_i) = \ell^{-}_i$; ``the most faithful description of
the data is the data''.\footnote{On the other extreme is a field with only
  $K=1$ predictive state, i.e.\ the iid case.}  As this tells us nothing about
the underlying dynamics, we must put some constraints on $K$ and/or $\epsilon$
to get a useful solution.  For now, assume that $K \ll N$ is fixed and we only
have to estimate $\epsilon$; in Section \ref{sec:merging}, we will give a
data-driven procedure to choose $K$.

\subsection{Nonparametric Likelihood Approximation}

To solve \eqref{eq:optim_problem_epsilon} with $K \ll N$ we need to evaluate
\eqref{eq:predictive_loglik_complete} for candidate solutions $\epsilon$.  We
cannot do this directly, since \eqref{eq:predictive_loglik_complete} involves
the unobserved $S_i$.  Moreover, predictive distributions can have arbitrary
shapes, so we want to use nonparametric methods, but this inhibits direct
evaluation of the component distributions $\Prob{X_i \mid \epsilon(\ell^{-}_i)
  = s_j }$.

We solve both difficulties together by using a nonparametric variant of the
expectation-maximization (EM) algorithm
\citep{Dempster-Laird-Rubin-introduce-EM}.  Following the recent nonparametric
EM literature \citep{Hall-et-al-nonparametric-mixtures,
  Benaglia-et-al-bandwidth-selection-in-EM-like,
  Bordes-et-al-stochastic-semiparametric-EM}, we approximate the $\Prob{X_i
  \mid \epsilon(\ell^{-}_i) = s_j }$ in the log-likelihood with kernel density
estimators (KDEs) using a previous estimate $\widehat{\epsilon}^{(n)}$.  That
is we approximate \eqref{eq:predictive_loglik_complete} with
$\widehat{\ell}^{(n)}(\epsilon; \mathcal{D}, S_i)$:
\begin{equation}
  \label{eq:loglik_full_nonparametric_approximation}
  \sum_{i=1}^{N}{\sum_{j=1}^{K}{\Indicator{S_i = s_j} \log  \widehat{f} \left( {X_i \mid \widehat{\epsilon}^{(n)}(\ell^{-}_i) = s_j } \right)}},
\end{equation}
where an equivalent version of $\widehat{f} \left( {X_i \mid
    \widehat{\epsilon}^{(n)}(\ell^{-}_i) = s_j } \right)$ is given below in
\eqref{eq:wKDE_state}.

\section{EM Algorithm for Predictive State Space Assignment}
\label{sec:EM}

Since $\epsilon$ maps to $\mathcal{S}$, the hidden state variable $S_i$ and the
``parameter'' $\epsilon$ play the same role. This in turn results in similar E
and M steps.  Figure \ref{fig:EM_algorithm_overview} gives an overview of the
proposed algorithm.

\subsection{Expectation Step}
The E-step requires the expected log-likelihood
\begin{equation}
  \label{eq:E-step_general}
  Q(\epsilon \mid \epsilon^{(n)}) = \E_{\mathcal{S} \mid \mathcal{D};
    \epsilon^{(n)}} \ell(\epsilon; \mathcal{D}, S_i),
\end{equation}
where expectation is taken with respect to $\Prob{S_i = s_j\mid \mathcal{D};
  \epsilon^{(n)} }$, the conditional distribution of the hidden variable $S_i$
given the data $\mathcal{D}$ and the current estimate $\epsilon^{(n)}$.  Using
\eqref{eq:predictive_loglik_complete} we obtain
\begin{align}
  \label{eq:E-step_mixedLICORS}
  \begin{split}
    Q(\epsilon \mid \epsilon^{(n)}) = \sum_{i=1}^{N} \sum_{j=1}^{K} & \,\Prob{S_i = s_j \mid X_i, \ell^{-}_i; \epsilon^{(n)}(\ell^{-}_i)} \\
    & \times \log \Prob{ X_i \mid \epsilon^{(n)}(\ell^{-}_i) = s_j }.
  \end{split}
\end{align}

As for $\ell(\epsilon; \mathcal{D})$ we use KDEs and obtain an approximate
expected log-likelihood
\begin{align}
  \label{eq:E-step_LCs}
  \begin{split}
    \widehat{Q}^{(n)}(\epsilon \mid \epsilon^{(n)}) = \sum_{i=1}^{N} \sum_{j=1}^{K} & \, \Prob{S_i = s_j \mid X_i, \ell^{-}_i; \epsilon^{(n)}(\ell^{-}_i)} \\
    & \times \log \widehat{f} \left( X_i \mid
      \widehat{\epsilon}^{(n)}(\ell^{-}_i) = s_j \right)
  \end{split}
\end{align}
The conditional distribution of $S_i$ given its FLC and PLC, $\lbrace X_i,
\ell^{-}_i \rbrace $, comes from Bayes' rule,
\begin{align}
  \Prob{S_i = s_j \mid X_i, \ell^{-}_i } & \propto \Prob{X_i, \ell^{-}_i \mid S_i = s_j } \Prob{ S_i = s_j }  \nonumber\\
  \label{eq:prob_S_given_data_after_bayes}
  & = \ProbState{j}{ X_i } \ProbState{j}{ \ell^{-}_i } \Prob{ S_i = s_j },
\end{align}
the second equality following from the conditional independence of $X_i$ and
$\ell^{-}_i$ given the state $S_i$.

For brevity, let $w_{ij} := \Prob{S_i = s_j \mid X_i, \ell^{-}_i }$, an $N
\times K$ weight matrix $\vecS{W}$, whose rows are probability distributions
over states.  This $\vecS{w}_i$ is the soft-thresholding version of
$\epsilon(\ell^{-}_i)$, so we can write the expected log-likelihood in terms of
$\vecS{W}$,
\begin{align}
  \label{eq:expected_loglik_W}
  \widehat{Q}^{(n)}(\vecS{W} \mid \widehat{\vecS{W}}^{(n)})&= \sum_{i=1}^{N}
  \sum_{j=1}^{K}{ w_{ij} \cdot \log \Prob{X_i \mid \widehat{\vecS{W}}^{(n)}_j}}
\end{align}

The current $\widehat{\vecS{W}}^{(n)}$ can be used to update (conditional)
probabilities in \eqref{eq:prob_S_given_data_after_bayes} by
\begin{align}
  \widehat{w}_{ij}^{(n+1)} \propto &
  \quad 
  \label{eq:probhat_S_given_data_after_bayes}
  \widehat{f}(x_i \mid S_i = s_j; \widehat{\vecS{W}}^{(n)}) \\
  & \times \mathcal{N}\left(\ell_i^{-} \mid \widehat{\boldsymbol \mu}_{j}^{(n)}, \widehat{\boldsymbol \Sigma}_{j}^{(n)}; \widehat{\vecS{W}}^{(n)}  \right) \\
  & \times \frac{\widehat{N}_j^{(n)}}{N},
\end{align}
where
\begin{inparaenum}[i)]
\item $\widehat{N}_j^{(n)} = \sum_{i=1}^{N} \widehat{\vecS{W}}^{(n)}_{ij}$ is
  the effective sample size of state $s_j$,
\item $\widehat{\boldsymbol \mu}_{j}^{(n)}$ and $\widehat{\boldsymbol
    \Sigma}_{j}^{(n)}$ are weighted mean and covariance matrix estimators of
  the PLCs using the $j$th column of $\widehat{\vecS{W}}^{(n)}$, and
\item the FLC distribution is estimated with a weighted\footnote{We also tried
    a hard-threshold estimator, but we found that the soft-threshold KDE
    performed better.} KDE (wKDE)
  \begin{equation}
    \label{eq:wKDE_state}
    \widehat{f}(x \mid S = s_j; \widehat{\vecS{W}}^{(n)})  = \frac{1}{\widehat{N}_j^{(n)}} \sum_{r=1}^{N} \widehat{\vecS{W}}_{rj}^{(n)} K_{h_j}(\vnorm{x_r - x}),
  \end{equation}
  Here the weights are again the $j$th column of $\widehat{\vecS{W}}^{(n)}$,
  and $K_{h_j}$ is a kernel function with a state-dependent bandwidth $h_j$.
  We used a Gaussian kernel, and to get a good, cluster-adaptive bandwidth
  $h_j$, we pick out on those $x_i$ for which $\argmax_k{w_{ik}} = j$
  (hard-thresholding of weights; cf.\
  \citealt{Benaglia-et-al-bandwidth-selection-in-EM-like}) and apply
  Silverman's rule-of-thumb-bandwidth (\texttt{bw.ndr0} in the R function
  \texttt{density}).
\end{inparaenum}
After estimation, we normalize each $\widehat{\vecS{w}}_i$ in
\eqref{eq:probhat_S_given_data_after_bayes}, $\widehat{w}_{ij}^{(n+1)}
\leftarrow
\frac{\widehat{w}_{ij}^{(n+1)}}{\sum_{j=1}^K{\widehat{w}_{ij}^{(n+1)}}}$.

Ideally, we would use a non-parametric estimate for the PLC distribution, e.g.,
forest density estimators
\citep{Liu-et-al-forest-density-estimation,Chow-Liu-trees}.  Currently,
however, such estimators are too slow to handle many iterations at large $N$,
so we model state-conditional PLC distributions as multivariate Gaussians.
Simulations suggest that this is often adequate in practice.

\subsection{Approximate Maximization Step}

In parametric problems, the M-step solves
\begin{equation}
\label{eq:M-step_general}
\epsilon^{(n+1)} = \argmax_{\epsilon}{\widehat{Q}^{(n)}(\epsilon \mid \epsilon^{(n)})},
\end{equation}
to improve the estimate.  Starting from a guess $\epsilon^{(0)}$, the EM
algorithm iterates \eqref{eq:E-step_general} and \eqref{eq:M-step_general}
to convergence.

In nonparametric problems, finding an $\epsilon^{(n+1)}$ that increases
$\widehat{Q}^{(n)}(\epsilon \mid \epsilon^{(n)})$ is difficult, since wKDEs
with non-zero bandwidth are not maximizing the likelihood; they are not even
guaranteed to increase it.  Optimizing \eqref{eq:E-step_LCs} by brute force
isn't computationally feasible either, as it would mean searching $K^N$ state
assignments (cf.\ \citealt{Bordes-et-al-stochastic-semiparametric-EM}).

\begin{figure}[!b]
  \begin{center}
    \fbox{
    \hspace{-0.4cm}
      \begin{minipage}{0.47\textwidth}
        \begin{enumerate}
        \setcounter{enumi}{-1}
      \item \label{step:initialize_theta} \textbf{Initialization:} Set $n =
        0$. Split data $\mathcal{D} = \lbrace X_i, \ell_i^{-}
        \rbrace_{i=1}^{N}$ in $\mathcal{D}_{train}$ and
        $\mathcal{D}_{test}$. Initialize states randomly from $\lbrace s_1,
        \ldots, s_K \rbrace$ $\rightarrow$ Boolean $\widehat{\vecS{W}}^{(0)}$.
        
      \item \label{step:E-step} \textbf{E-step:} Obtain updated
        $\widehat{\vecS{W}}^{(n+1)}$ via
        \eqref{eq:probhat_S_given_data_after_bayes}.
      \item \label{step:M-step} \textbf{Approximate M-step:} Update mixture
        pdfs $\Prob{x_i \mid S_i = s_j}$ via \eqref{eq:wKDE_state} with
        $\widehat{\vecS{W}}^{(n+1)}$.
      \item \label{step:prediction-step} \textbf{Out-of-sample Prediction:}
        Evaluate out-of-sample MSE for $\widehat{\vecS{W}}^{(n+1)}$ by
        predicting FLCs from PLCs in $\mathcal{D}_{test}$. Set $n = n + 1$.
      \item \label{step:temp_convergence} \textbf{Temporary convergence:}
        Iterate \ref{step:E-step} - \ref{step:prediction-step} until
        \begin{equation}
          \vnorm{ \widehat{\vecS{W}}^{(n)} - \widehat{\vecS{W}}^{(n-1)} } < \delta 
        \end{equation}
      \item \label{step:merging-step} \textbf{Merging:} Estimate pairwise
        distances (or test)
        \begin{equation}
          \label{eq:distance_between_state_dists}
          \widehat{d}_{jk} = \operatorname{dist}\left( \widehat{f}^{(n)}_j, \widehat{f}^{(n)}_k \right) \forall j,k = 1, \ldots, K.
        \end{equation} 
        \begin{enumerate}
        \item If $K > 1$: determine $(j^{(\min)}, k^{(\min)}) = \argmin_{j\neq
            k} \widehat{d}_{jk}$ and merge these columns 
          \begin{equation}
            \vecS{W}^{(n)}_{j^{(\min)}} \leftarrow \vecS{W}^{(n)}_{j^{(\min)}} + \vecS{W}^{(n)}_{k^{(\min)}}
          \end{equation}
          Omit column $k^{(\min)}$ from $\vecS{W}^{(n)}_{k^{(\min)}}$, set $K =
          K - 1$, and re-start iterations at \ref{step:E-step}.
        \item If $K = 1$: return $\widehat{\vecS{W}}^{*}$ with lowest
          out-of-sample MSE.
        \end{enumerate}
      \end{enumerate}
    \end{minipage}
  }
\end{center}
\caption[Outline of mixed LICORS: nonparametric EM algorithm for predictive
state estimation]{\label{fig:EM_algorithm_overview} Mixed LICORS: nonparametric
  EM algorithm for predictive state recovery.}
\end{figure}

However, in our particular setting the parameter space and the expectation of
the hidden variable coincide, since $\widehat{\vecS{w}}_i$ is a
soft-thresholding version of $\epsilon(\ell^{-}_i)$.  Furthermore, none of the
estimates above requires a deterministic $\epsilon$ mapping; they are all
weighted MLEs or KDEs.  Thus, like
\citet{Bengalia-et-al-EM-like-nonparametric}, we take the weights from the
E-step, $\widehat{\vecS{W}}^{(n+1)}$, to update each component distribution
using \eqref{eq:wKDE_state}. This in turn can then be plugged into
\eqref{eq:loglik_full_nonparametric_approximation} to update the likelihood
function, and in \eqref{eq:probhat_S_given_data_after_bayes} for the next
E-step.

The wKDE update does not solve \eqref{eq:M-step_general} nor does it provably
increase the log-likelihood (although in simulations it often does so).  We
thus use cross-validation (CV) to select the best $\widehat{\vecS{W}}^{*}$, and
henceforth do not rely on an ever-increasing log-likelihood as the usual
stopping rule in EM algorithms (see Section \ref{sec:simulations} for details).

\subsection{Data-driven Choice of $K$: Merge States To Obtain Minimal
  Sufficiency}
\label{sec:merging}

One advantage of the mixture model in
\eqref{eq:joint_pdf_product_of_pred_states} is that predictive states have, by
definition, comparable conditional distributions.  Since conditional densities
can be tested for equality by a nonparametric two-sample test (or using a
distribution metric), we can merge nearby classes. We thus propose a
data-driven automatic selection of $K$, which solves this key challenge in
fitting mixture models:
\begin{inparaenum}[1)]
\item start with a sufficiently large number of clusters, $K_{\max} < N$;
\item test for equality of predictive distribution each time the EM reaches a
  (local) optimum;
\item merge until $K=1$ (iid case) -- step \ref{step:merging-step} in Fig.\
  \ref{fig:EM_algorithm_overview};
\item choose the best model $\widehat{\mathbf{W}}^*$ (and thus $K^*$) by CV.
\end{inparaenum}

\section{Forecasting Given New Data}
\label{sec:prediction_mixedLICORS}

The estimate $\widehat{\vecS{W}}^{*}$ can be used to forecast $\tilde{X} $
given a new $\tilde{\ell}^{-}$.  Integrating out $S_i$ yields a mixture
\begin{align}
\label{eq:predictive_distribution_new_PLC}
\Prob{\tilde{X} \mid \tilde{\ell}^{-} } & = \sum_{j=1}^{K} \Prob{\tilde{S} = s_j \mid \tilde{\ell}^{-} } \cdot  \ProbState{j}{\tilde{X}}.
\end{align}
As $\ProbState{j}{\tilde{X}} = \Prob{\tilde{X} = x \mid \tilde{S} = s_j }$ is
independent of $\tilde{\ell}^{-}$ we do not have to re-estimate them for each $
\tilde{\ell}^{-}$, but can use the wKDEs in \eqref{eq:wKDE_state} from the
training data.

The mixture weights $\tilde{w}_j := \Prob{\tilde{S} = s_j \mid \tilde{\ell}^{-}
}$ are in general different for each PLC and can again be estimated using
Bayes's rule (with the important difference that now we only condition on
$\tilde{\ell}^{-}$, not on $\tilde{X}$):
\begin{align}
\label{eq:update_weights_Bayes_PLC_only_mixedLICORS}
\widehat{\vecS{w}}_j(\tilde{\ell}^{-}; \widehat{\vecS{W}}^{*})
 & \propto \ProbHat{\tilde{\ell}^{-} \mid \tilde{S} = s_j; \widehat{\vecS{W}}^{*}} \times \ProbHat{ \tilde{S} = s_j; \widehat{\vecS{W}}^{*} } \nonumber \\
& = \mathcal{N}\left(\tilde{\ell}^{-}; \widehat{\boldsymbol \mu}_{(j)}^*, \widehat{\boldsymbol \Sigma}_{(j)}^*; \widehat{\vecS{W}}^{*}  \right) \times \frac{\widehat{N}_j^{*}}{N}.
\end{align}
After re-normalization of $\widehat{\tilde{\vecS{w}}} = (
\widehat{\tilde{w}}_1^{*}, \ldots, \widehat{\tilde{w}}_K^{*} )$, the predictive
distribution \eqref{eq:predictive_distribution_new_PLC} can be estimated via
\begin{align}
\label{eq:estimated_predictive_distribution_new_PLC}
\ProbHat{\tilde{X} = x \mid \tilde{\ell}^{-} } & = \sum_{j=1}^{K} \widehat{\tilde{w}}_j^{*} \cdot \widehat{f}\left(\tilde{X} = x \mid S_i = s_j; \widehat{\vecS{W}}^{(*)}\right).
\end{align}

A point forecast can then be obtained by a weighted combination of point
estimates in each component (e.g.\ weighted mean), or by the mode of the full
distribution.  In the simulations we use the weighted average from each
component as the prediction of $\tilde{X}$.

\section{Simulations}
\label{sec:simulations}

To evaluate the
predictive ability of mixed LICORS in a practical, non-asymptotic context 
we use the following simulation.
The continuous-valued $(1+1)D$ field $\field{X}{r}{t}$ has a discrete latent
state space $\field{d}{r}{t}$. The observable field $\field{X}{r}{t}$ evolves
according to a conditional Gaussian distribution,
\small
\begin{align}
  \label{eq:CA_cont_predictive_rule_normal}
  \Prob{\field{X}{r}{t} \mid \field{d}{r}{t}} & = 
  \begin{cases}
    \mathcal{N}(\field{d}{r}{t}, 1),& \text{ if } | \field{d}{r}{t} | < 4, \\
    \mathcal{N}(0, 1), & \text {otherwise},
  \end{cases} \\
  \label{eq:CA_cont_predictive_rule_normal_initial_conditions}
  \text {and initial conditions: } & \field{X}{\cdot}{1} = \field{X}{\cdot}{2}
  = \boldsymbol 0 \in \R^{\card{S}}.
\end{align}
\normalsize
 The state space $\field{d}{r}{t}$ evolves with the observable
field,
\begin{align}
  \begin{split}
  \label{eq:state_description}
  \field{d}{r}{t}  = & \left[ \frac{\sum_{i=-2}^{2}{\field{X}{(r+i) \bmod \card{S}}{t-2}}}{5} \right. \\
  & - \left.\frac{\sum_{i=-1}^{1}{\field{X}{(r+i) \bmod \card{S}}{t-1}}}{3} \right]
  \end{split}
\end{align}
where $[x]$ is the closest integer to $x$.  In words, Eq.\
\eqref{eq:state_description} says that the latent state $\field{d}{r}{t}$ is
the rounded difference between the sample average of the $5$ nearest sites at
$t-2$ and the sample average of the $3$ nearest sites at $t-1$.  Thus $h_p = 2$
and $c=1$.

If we include the present in the FLC, \eqref{eq:CA_cont_predictive_rule_normal}
gives $h_f = 0$, making FLC distributions one-dimensional.  As
$\field{d}{r}{t}$ is integer-valued and the conditional distribution becomes
$\mathcal{N}(0,1)$ if $ | \field{d}{r}{t} | > 4$, the system has $7$ predictive
states, $\lbrace s_{-3}, s_{-2}, \ldots, s_{2}, s_{3} \rbrace$, distinguished
by the conditional mean $\E \left(\field{X}{r}{t} \mid s_k \right) = k$. Thus
$\field{X}{r}{t} \mid s_k = \mathcal{N}(k, 1)$, $k = -3, -2, \ldots, 2, 3$.

\begin{figure}[!t]
  \centering
  \begin{subfigure}[t]{0.49\textwidth}
    \centering
    \includegraphics[width=\textwidth]{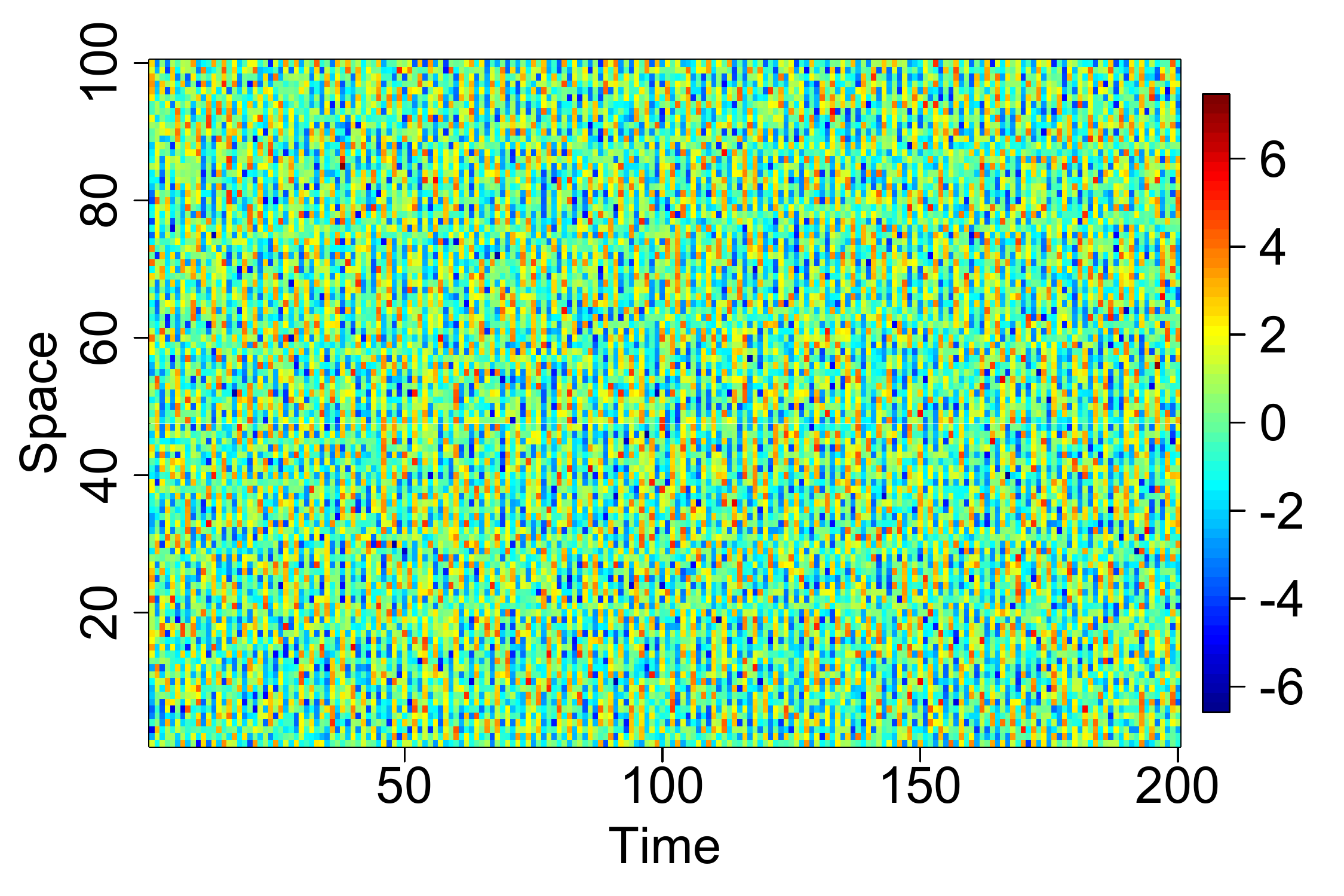}
    \caption{\label{fig:cont_CA_round_normal_right_oriented_observed} observed
      field $\field{X}{r}{t}$}
  \end{subfigure}%
  
  \begin{subfigure}[t]{0.49\textwidth}
    \centering
    \includegraphics[width=\textwidth]{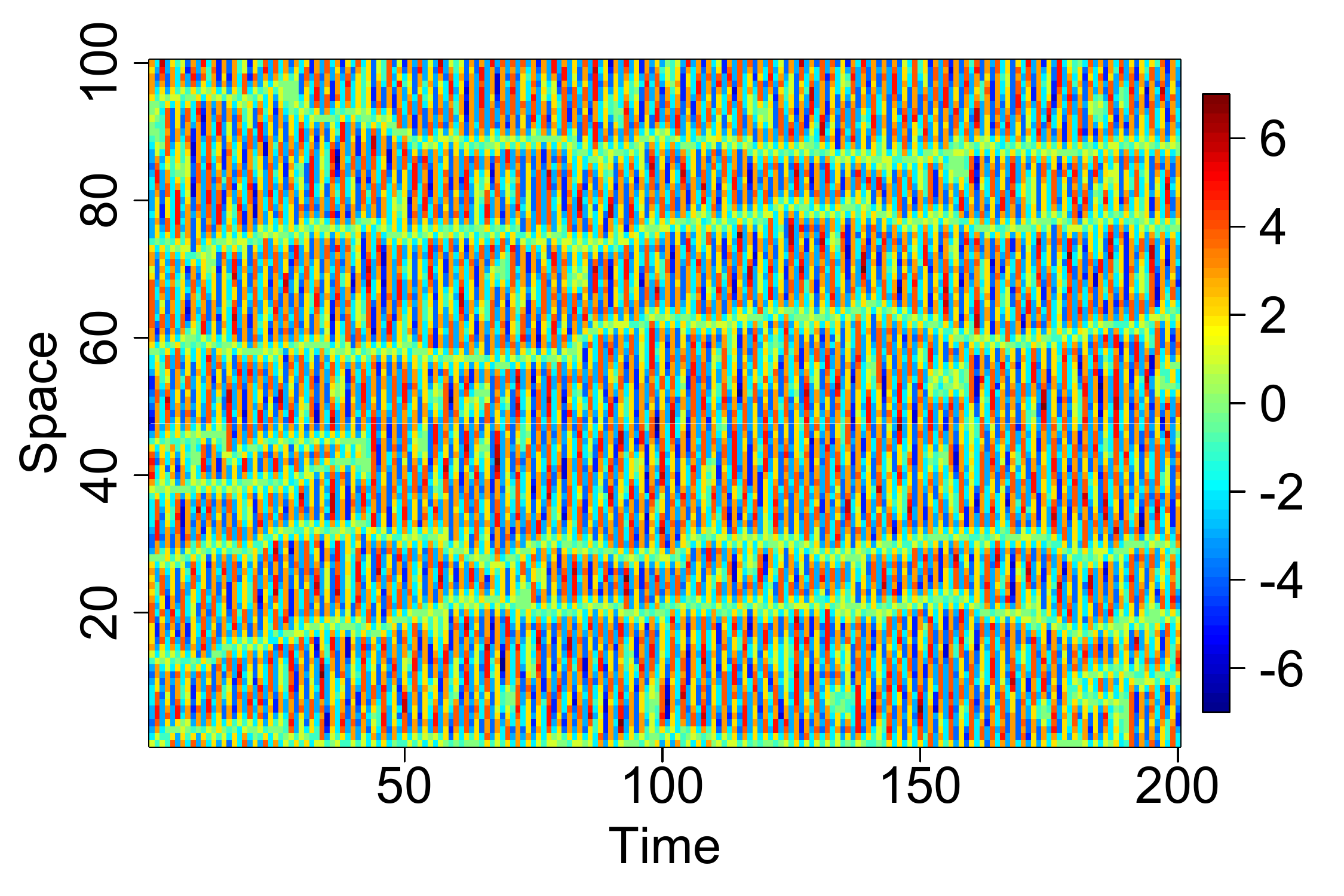}
    \caption{\label{fig:cont_CA_round_normal_right_oriented_states} state-space
      $\field{d}{r}{t}$}
  \end{subfigure}
  \caption{\label{fig:cont_CA_round_states_Simulations} Simulation of
    \eqref{eq:CA_cont_predictive_rule_normal} -- \eqref{eq:state_description}.}
\end{figure}

\begin{figure}[!t]
  \centering
  \includegraphics[width=.49\textwidth]{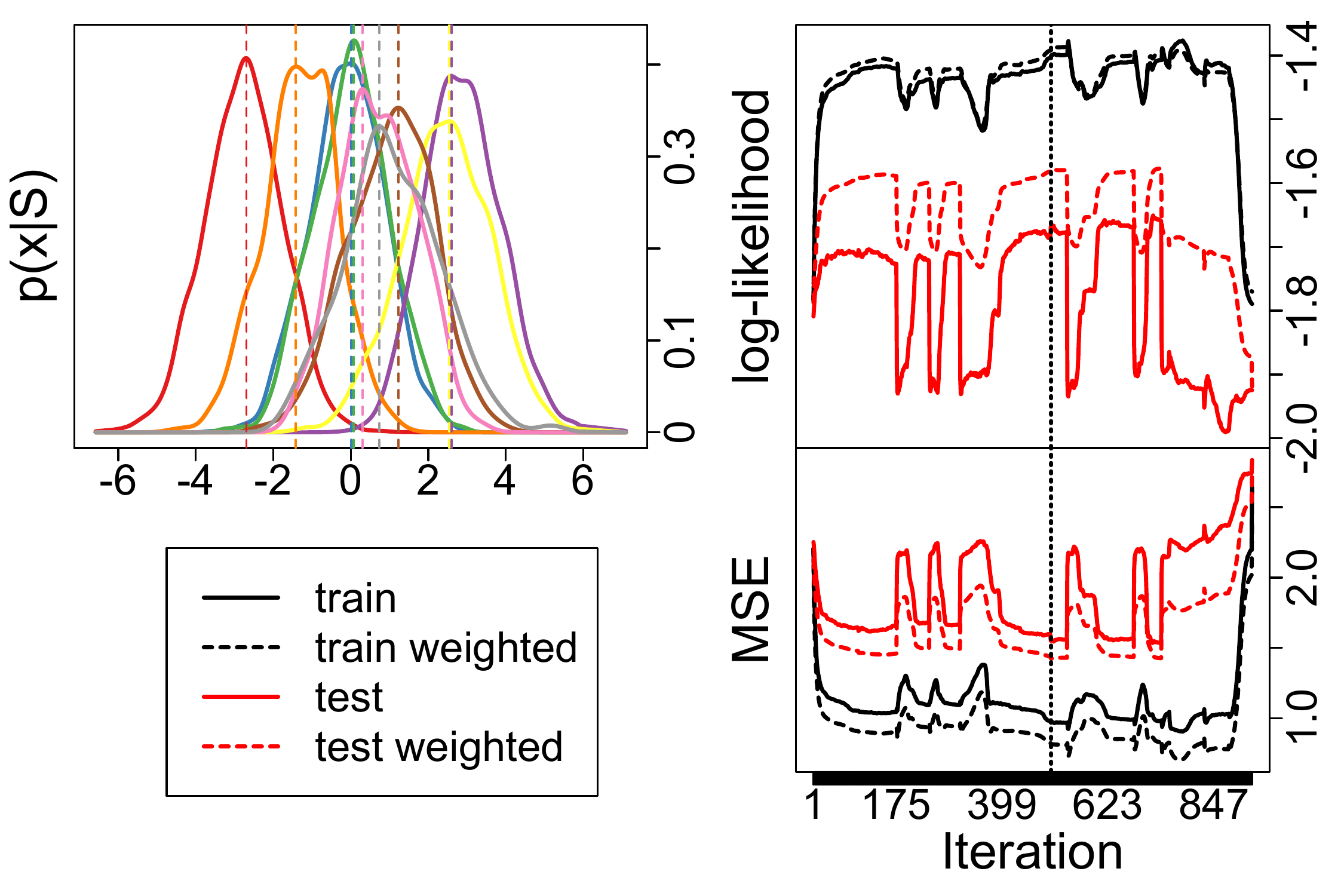}
  \caption[Summary of mixed LICORS estimates]{Mixed LICORS with $h_p = 2$ and
    $K_{\max} = 15$ starting states. Nonparametric estimates of 
    conditional predictive distributions $\Prob{X = x \mid S = s_j}$ (top-left); 
    trace plots for log-likelihood and MSE (right).}
  \label{fig:AISTATS_cont_CA_round_normal_mixedLICORS_analytics_none_sparsity_lambda_0}
\end{figure}

Figure \ref{fig:cont_CA_round_states_Simulations} shows one realization of
\eqref{eq:CA_cont_predictive_rule_normal} -- \eqref{eq:state_description} for
$\vecS{S} = \lbrace 1, \ldots, 100 \rbrace$ (vertical) and $t = 1, \ldots, T =
200$ (left to right), where we discarded the first $100$ time steps as burn-in
to avoid too much dependence on the initial conditions
\eqref{eq:CA_cont_predictive_rule_normal_initial_conditions}. While the
(usually unobserved) state space has distinctive green temporal traces and also
alternating red and blue patches, the observed field is too noisy to clearly
see any of these patterns.\footnote{All computation was done in R \citep{R10}.}

\begin{figure*}[!t]
  \centering
  \begin{subfigure}[t]{0.32\textwidth}
    \centering
    \includegraphics[width=\textwidth]{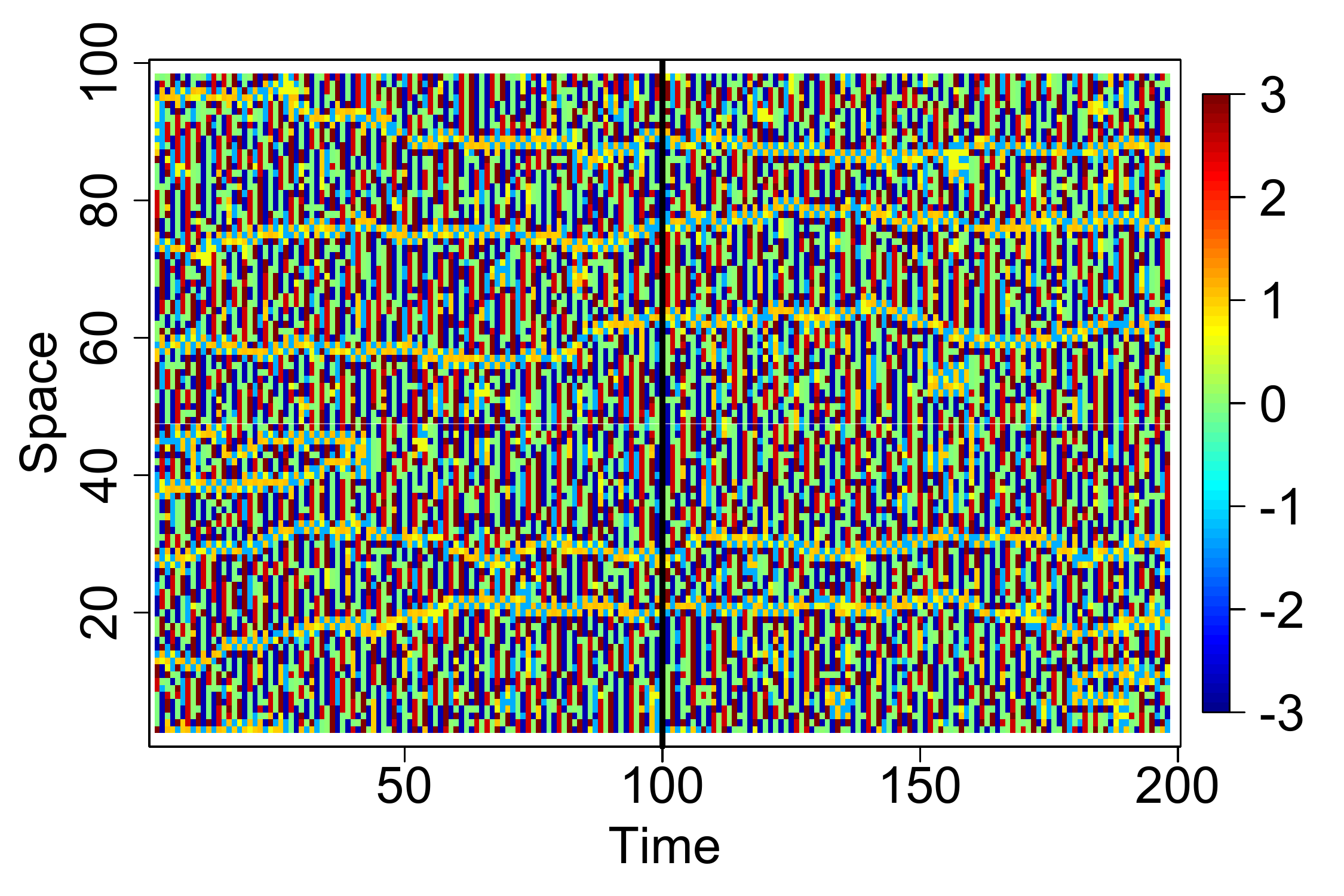}
    \caption{\label{fig:cont_CA_round_normal_right_oriented_predictions_deterministic_sparsity_mixedLICORS}
      estimated predictive states
      $\field{\widehat{S}}{r}{t}$} 
  \end{subfigure}
  \begin{subfigure}[t]{0.32\textwidth}
    \centering
    \includegraphics[width=\textwidth]{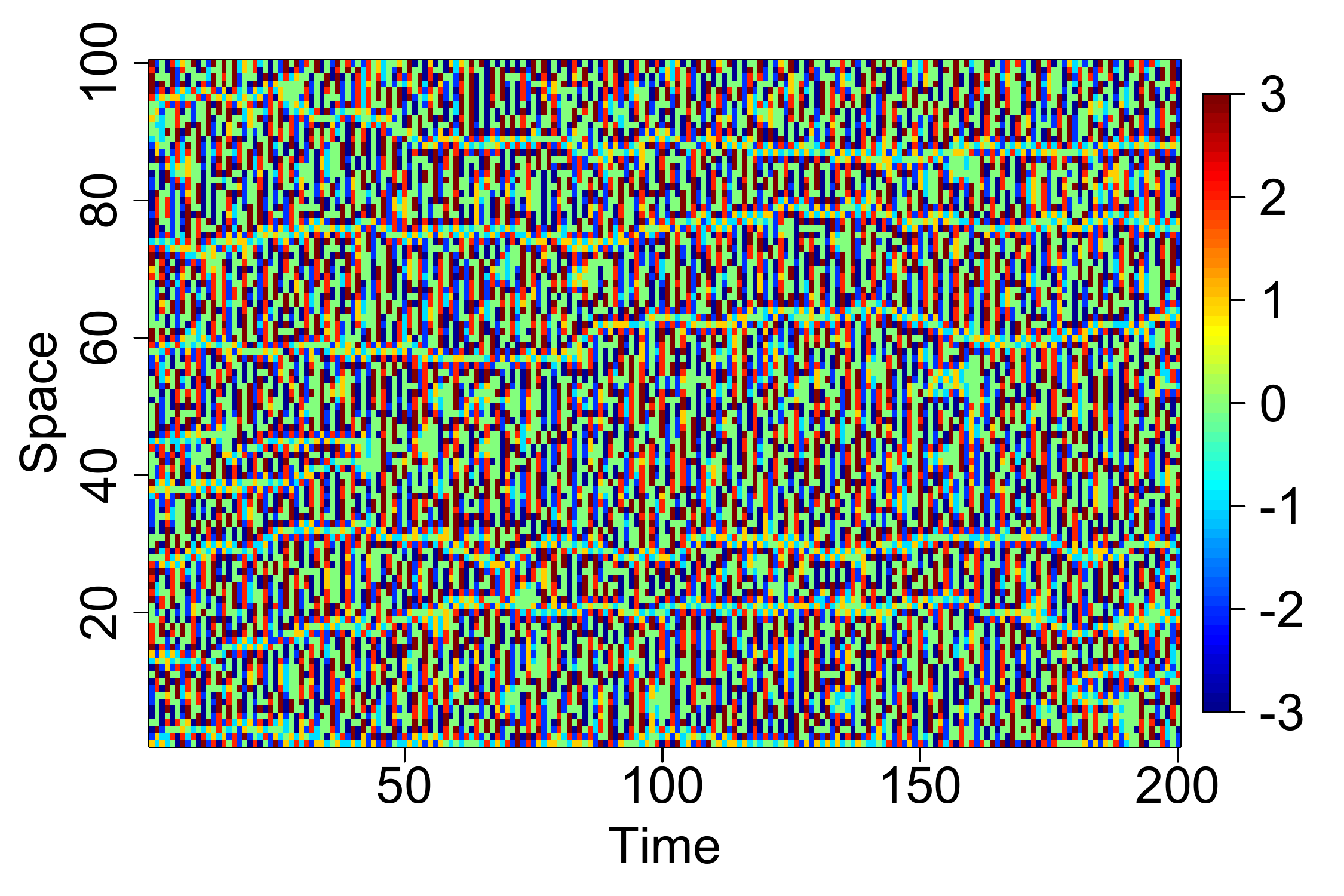}
    \caption{\label{fig:cont_CA_round_normal_right_oriented_predictive_states}
      true predictive states $\field{S}{r}{t}$}
  \end{subfigure}
  \begin{subfigure}[t]{0.32\textwidth}
    \centering
    \includegraphics[width=\textwidth]{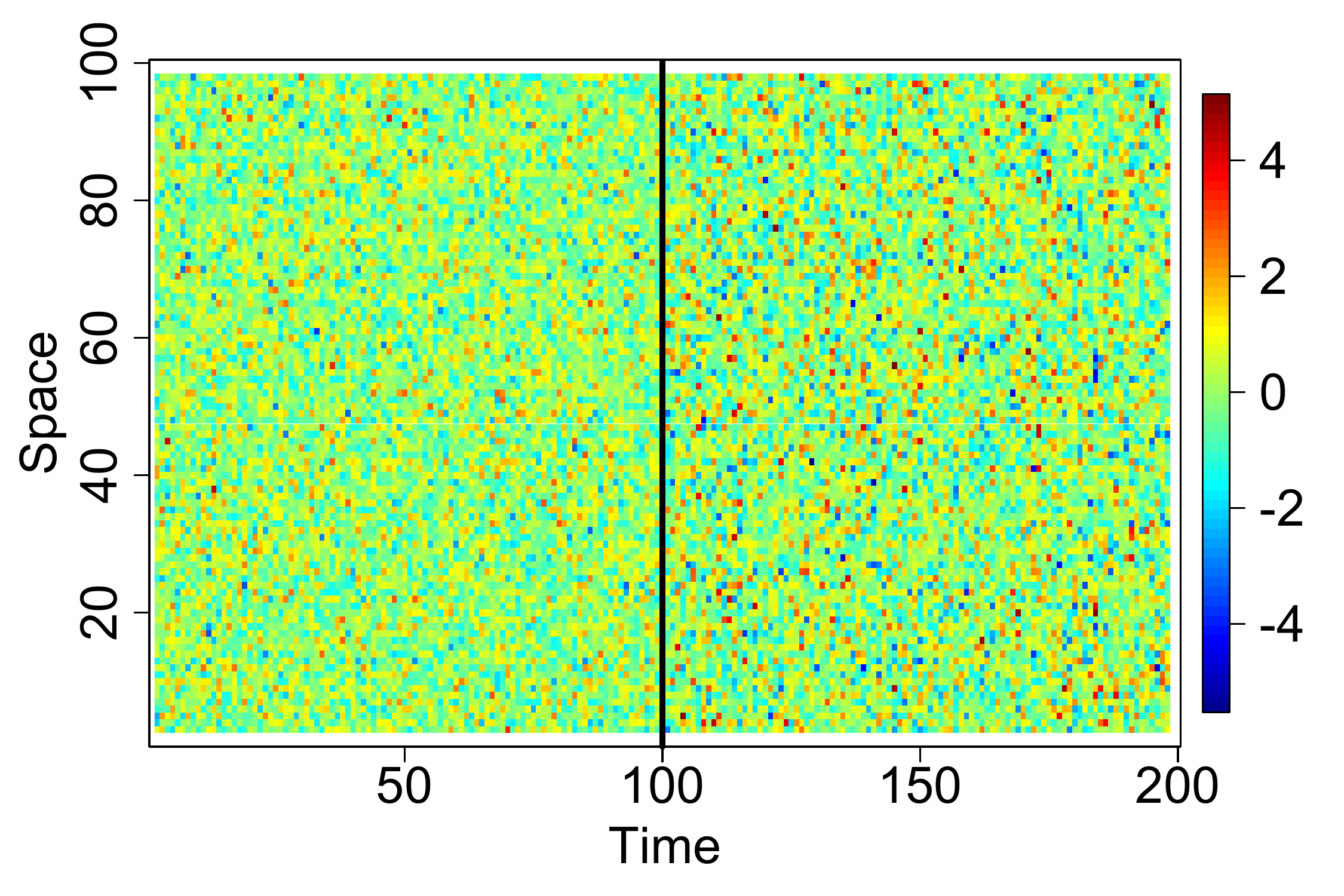}
    \caption{\label{fig:cont_CA_round_normal_right_oriented_residuals_deterministic_sparsity_mixedLICORS}
      residuals}
  \end{subfigure}
  \caption[Mixed LICORS model check: (a) truth; (b) estimate; (c)
  residuals.]{\label{fig:cont_CA_round_states_LICORS} Mixed LICORS model fit
    and residual check.}
\end{figure*}

Figure \ref{fig:AISTATS_cont_CA_round_normal_mixedLICORS_analytics_none_sparsity_lambda_0} summarizes one run of mixed LICORS with $K = 15$
initial states, $h_p = 2$, and $L_1$ distance in \eqref{eq:distance_between_state_dists}. The first $100$ time steps were used as training
data, and the second half as test data. The optimal $\widehat{\vecS{W}}^*$,
which minimized the out-of-sample weighted MSE, occurred at iteration $502$,
with $\widehat{K} = 9$ estimated predictive states.  The trace plots 
show large temporary drops (increases) in the log-likelihood (MSE) whenever 
the EM reaches a local optimum and merges two states. After
merging, the forecasting performance and log-likelihood quickly return to ---
or even surpass --- previous optima. 

The predictions from $\widehat{\vecS{W}}^*$ in Fig.\
\ref{fig:cont_CA_round_normal_right_oriented_predictions_deterministic_sparsity_mixedLICORS}
show that mixed LICORS is practically unbiased -- compare to the visually
indistinguishable true state space in Fig.\
\ref{fig:cont_CA_round_normal_right_oriented_predictive_states}.  The residuals
in Fig.\
\ref{fig:cont_CA_round_normal_right_oriented_residuals_deterministic_sparsity_mixedLICORS}
show no obvious patterns except for a larger variance in the right half
(training vs.\ test data).

\subsection{Mixed versus Hard LICORS}

Mixed LICORS does better than hard LICORS at forecasting. We use $100$
independent realizations of \eqref{eq:CA_cont_predictive_rule_normal} --
\eqref{eq:state_description} and for each one we train the model on the first
half, and test it on the second (future) half. Lower out-of-sample future MSEs
are, of course, better.

We use EM as outlined in Fig.\ \ref{fig:EM_algorithm_overview} with $K_{\max} =
15$ states, $10^3$ maximum number of iterations. We keep the estimate
$\widehat{\vecS{W}}^*$ with the lowest out-of-sample MSE over $10$ independent
runs. The first run is initialized with a K-means++ clustering
\citep{k-means-plus-plus} on the PLC space; state initializations in the
remaining nine runs were uniformly at random from $\mathcal{S} = \lbrace s_1,
\ldots, s_{15} \rbrace$.

To test whether mixed LICORS accurately estimates the mapping $\epsilon:
\ell^{-} \mapsto \mathcal{S}$, we also predict FLCs of an independently
generated realization of the same process.  If the out-of-sample MSE for the
independent field is the same as the out-of-sample MSE for the future evolution
of the training data, then mixed LICORS is not just memorizing fluctuations in
any given realization, but estimates characteristics of the random field.  We
calculated both weighted-mixture forecasts, and the forecast of the state with
the highest weight.

\begin{figure}[!b]
  \centering
  \includegraphics[width=.48\textwidth, trim = 0cm 0cm 0.5cm 0.5cm, clip =
  true]{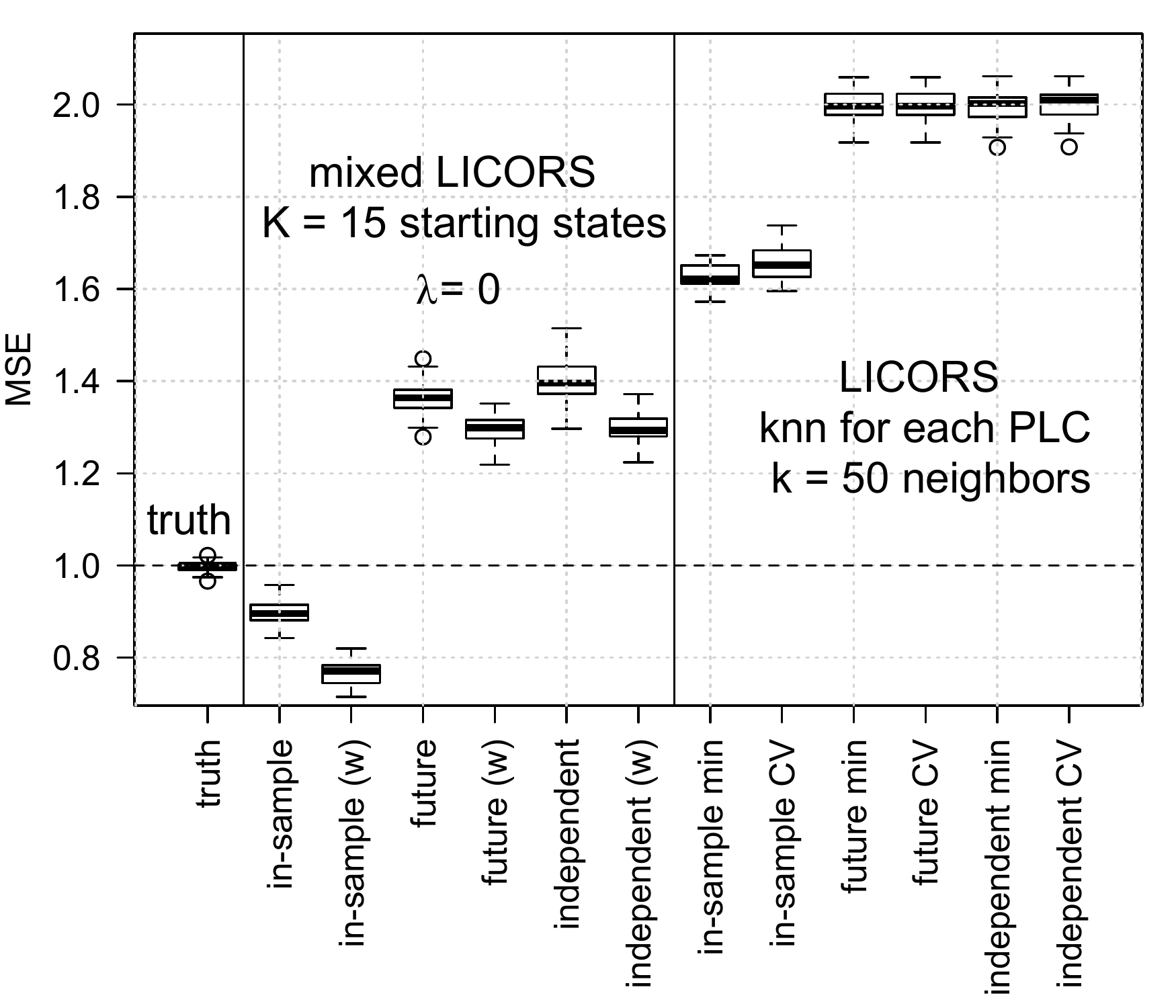}
  \caption{\label{fig:ContCA_h_p=2_sd1_trunc4_outofsampleCV_MSE_mixedLICORS_none_sparsity_alpha0_1_startstates15_boxplot}
    Comparing mixed and hard LICORS by forecast MSEs.}
\end{figure}

Figure
\ref{fig:ContCA_h_p=2_sd1_trunc4_outofsampleCV_MSE_mixedLICORS_none_sparsity_alpha0_1_startstates15_boxplot}
shows the results. Mixed LICORS greatly improves upon the hard LICORS
estimates, with up to $34\%$ reduction in out-of-sample MSE. As with hard
LICORS, the MSE for future and independent realizations are essentially the
same, showing that mixed LICORS generalizes from one realization to the whole
process.  As the weights often already are $0$/$1$ assignments,
weighted-mixture prediction is only slightly better than using the
highest-weight state.

\subsection{Simulating New Realizations of The Underlying System}

Recall that all simulations use
\eqref{eq:CA_cont_predictive_rule_normal_initial_conditions} as initial
conditions.  If we want to know the effect of different starting conditions,
then we can simply use Eqs.\ \eqref{eq:CA_cont_predictive_rule_normal} \&
\eqref{eq:state_description} to simulate that process, since they fully specify
the evolution of the stochastic process. In experimental studies, however,
researchers usually lack such generative models; learning one is often the
point of the experiment in the first place.

Since mixed LICORS estimates joint and conditional predictive distributions,
and not only the conditional mean, it is possible to simulate a new realization
from an estimated model.  Figure \ref{fig:simulation_overview} outlines this
simulation procedure.  We will now demonstrate that mixed LICORS can be used
instead to simulate from different initial conditions \emph{without} knowing
Eqs.\ \eqref{eq:CA_cont_predictive_rule_normal} \&
\eqref{eq:state_description}.

\begin{figure}[!t]
  \begin{center}
    \fbox{
    \hspace{-0.4cm}
      \begin{minipage}{0.47\textwidth}
        \begin{enumerate}
        \setcounter{enumi}{-1}
        \item Initialize field $\lbrace \field{X}{\cdot}{1-\tau} \rbrace_{\tau = 1}^{h_p}$. Set $t = 1$.
        \item \label{item:fetch} Fetch all PLCs at time $t$: $P_{t} = \lbrace \field{\ell^{-}}{r}{t} \rbrace_{\vecS{r} \in \vecS{S}^{\operatorname{new}}}$
        \item For each $\field{\ell^{-}}{r}{t} \in P_{t}$:
        \begin{enumerate}
        \item \label{step:draw_state} draw state $s_j \sim \Prob{S = s_j \mid  \field{\ell^{-}}{r}{t}}$ 
        \item \label{step:draw_observation_given_state} draw $\field{X}{r}{t} \sim \Prob{x \mid S = s_j}$ 
        \end{enumerate}
        \item If $t < t_{\max}$, set $t = t+1$ and go to step \ref{item:fetch}. Otherwise return simulation $\lbrace \field{X}{\cdot}{t} \rbrace_{t = 1}^{t_{\max}}$.
        \end{enumerate}
      \end{minipage}
    }
  \end{center}
  \caption[Simulating new realization from estimated
  model]{\label{fig:simulation_overview} Simulate new realization from
    spatio-temporal process on $\vecS{S}^{\operatorname{new}} \times \lbrace 1,
    \ldots, t_{\max} \rbrace$ using true and estimated dynamics (use
    \eqref{eq:update_weights_Bayes_PLC_only_mixedLICORS} in
    \ref{step:draw_state}; \eqref{eq:wKDE_state} in
    \ref{step:draw_observation_given_state}).}
\end{figure}

\begin{figure}[!t]
        \centering
        \begin{subfigure}[t]{0.49\textwidth}
                \centering
                \includegraphics[width=\textwidth]{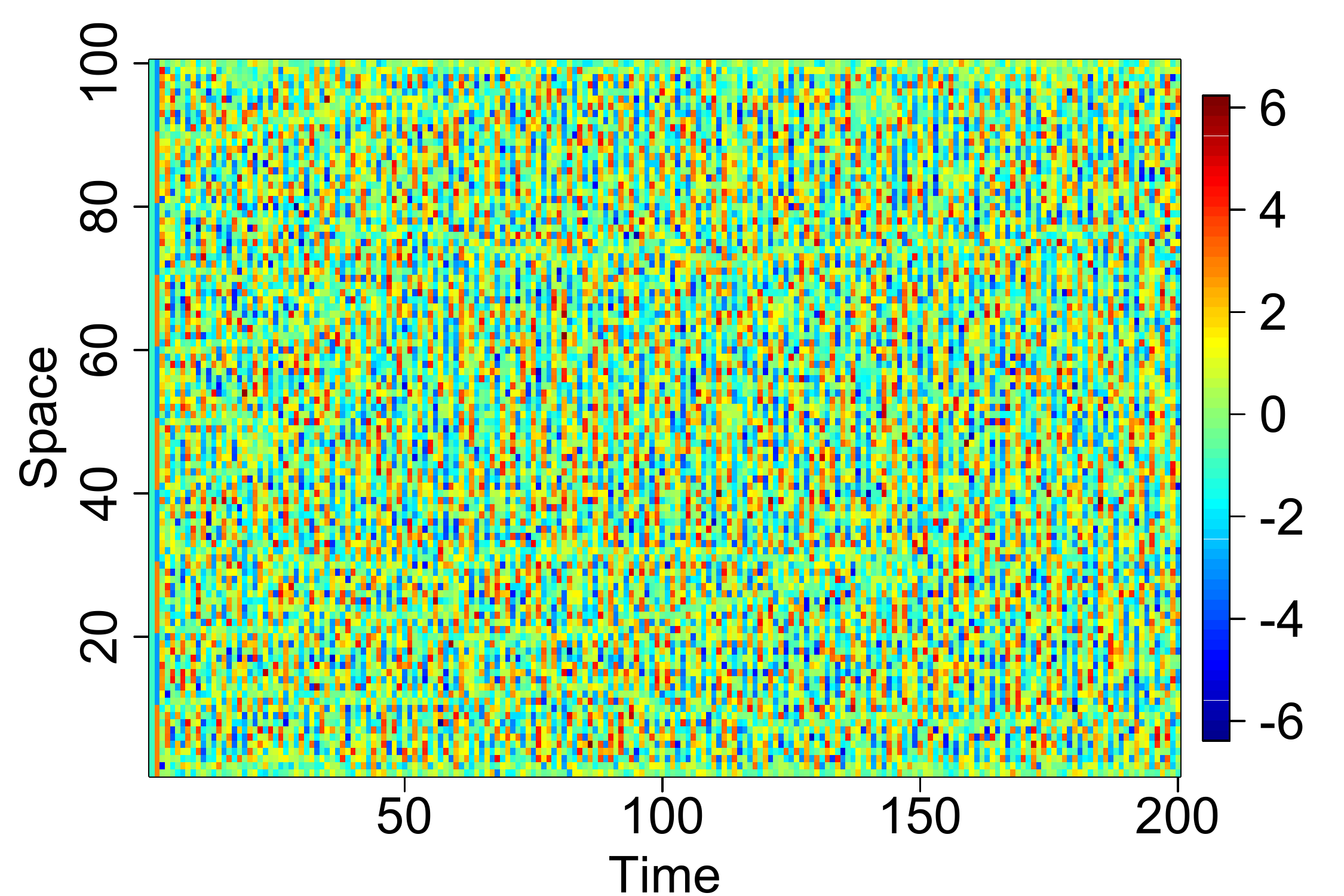}
                \caption{\label{fig:cont_CA_round_normal_right_oriented_simulated_data_true_dynamics} simulations from true dynamics}
        \end{subfigure}%

        \begin{subfigure}[t]{0.49\textwidth}
                \centering
                \includegraphics[width=\textwidth]{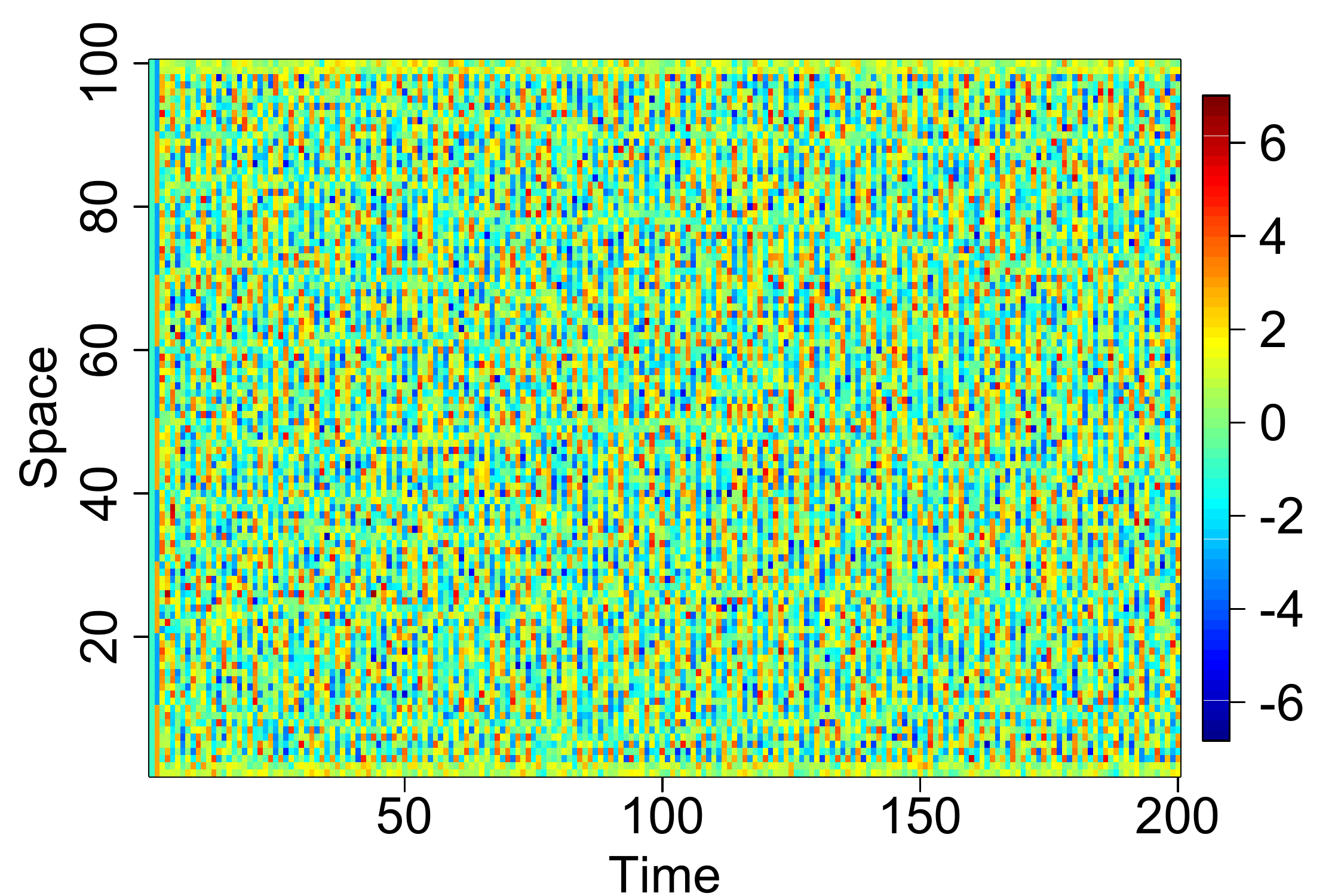}
                \caption{\label{fig:cont_CA_round_normal_right_oriented_simulated_data_deterministic_sparsity_mixedLICORS} simulations from estimated dynamics}
        \end{subfigure}
        

        
        \caption[Mixed LICORS simulations: (a) from the estimated model; (b) from the truth.]{\label{fig:cont_CA_round_normal_simulations} Simulating another realization of \eqref{eq:CA_cont_predictive_rule_normal} \& \eqref{eq:state_description} but with different starting conditions.}
\end{figure}

For example, Fig.\
\ref{fig:cont_CA_round_normal_right_oriented_simulated_data_true_dynamics}
shows a simulation
using the true mechanisms in \eqref{eq:CA_cont_predictive_rule_normal} \&
\eqref{eq:state_description} with starting conditions $\field{X}{\cdot}{1} =
-\boldsymbol 1$ and $\field{X}{\cdot}{2} = \pm \boldsymbol 3 \in
\R^{\card{\vecS{S}}}$ in alternating patches of ten times $3$, ten times $-3$,
ten times $3$, etc. (total of $10$ patches since $\card{\vecS{S}} = 100$). The
first couple of columns (on the left) are influenced by different starting
conditions, but the initial effect dies out soon (since $h_p = 2$) and similar
structures (left to right traces) as in simulations with
\eqref{eq:CA_cont_predictive_rule_normal_initial_conditions} emerge (Fig.\
\ref{fig:cont_CA_round_states_Simulations}).

Figure
\ref{fig:cont_CA_round_normal_right_oriented_simulated_data_deterministic_sparsity_mixedLICORS}
shows simulations solely using the mixed LICORS estimates in Fig.\
\ref{fig:AISTATS_cont_CA_round_normal_mixedLICORS_analytics_none_sparsity_lambda_0}.  While the patterns are quantitatively different (due to
random sampling), the qualitative structures are strikingly similar.  Thus mixed \method can not only accurately estimate
$\field{S}{r}{t}$, but also learn the optimal prediction rule
\eqref{eq:CA_cont_predictive_rule_normal} solely from the observed data
$\field{X}{r}{t}$.

\section{Discussion}
\label{sec:discussion}

Mixed LICORS attacks the problem of reconstructing the predictive state space
of a spatio-temporal process through a nonparametric, EM-like algorithm, softly
clustering observed histories by their predictive consequences.  Mixed LICORS
is a probabilistic generalization of hard LICORS and can thus be easily adapted
to other statistical settings such as classification or regression. Simulations
show that it greatly outperforms its hard-clustering predecessor.

However, like other state-of-the-art nonparametric EM-like algorithms
\citep{Hall-et-al-nonparametric-mixtures,
  Bordes-et-al-stochastic-semiparametric-EM,
  Mallapragada-et-al-non-parametrix-mixture}, theoretical properties of our
procedure are not yet well understood.  In particular, the nonparametric
estimation of mixture models poses identifiability problems \citep[\S 2 and
references therin]{Bengalia-et-al-EM-like-nonparametric}.  Here, we
demonstrated that in practice mixed LICORS does not suffer from identifiability
problems, and outperforms (identifiable) hard-clustering methods.

We also demonstrate that mixed LICORS can learn spatio-temporal dynamics from
data, which can then be used for simulating new experiments, whether from the
observed initial conditions or from new ones.  Simulating from observed
conditions allows for model checking; simulating from new ones makes
predictions about unobserved behavior.  Thus mixed LICORS can in principle make
a lot of expensive, time- and labor-intensive experimental studies much more
manageable and easier to plan.  In particular, mixed LICORS can be applied to
e.g., functional magnetic resonance imaging (fMRI) data to analyze and forecast
complex, spatio-temporal brain activity. Due to space limits we refer to future
work.

\subsubsection*{Acknowledgments}
We thank Stacey Ackerman-Alexeeff, Dave Albers, Chris Genovese, Rob Haslinger, Martin Nilsson Jacobi, Heike Jänicke, Kristina Klinkner, Cristopher Moore, Jean-Baptiste Rouquier, Chad Schafer, Rafael Stern, and Chris Wiggins for valuable discussion, and Larry Wasserman for detailed suggestions that have improved all aspects of this work. GMG was supported by an NSF grant (\# DMS 1207759). CRS was supported by grants from INET and from the NIH (\# 2 R01 NS047493). 

\clearpage

\addcontentsline{toc}{subsubsection}{References}
\bibliographystyle{abbrvnat}
\bibliography{../../../bib/PhD_thesis,../../../bib/locusts}

\cleardoublepage

\appendix

\begin{center}
\textbf{ APPENDIX -- SUPPLEMENTARY MATERIAL }
\end{center}

\section{Proofs}
\label{sec:proofs}

Here we will show that the joint pdf of the entire field conditioned on its
margin (Figure \ref{fig:field_margin}, gray area) equals the product of the
predictive conditional distributions.

\begin{proof}[Proof of Proposition \ref{prop:joint_factorizes}]
\label{proof:prop:joint_factorizes}

  To simplify notation, we assign index numbers $i \in 1, \ldots N$ to the
  space-time grid $\field{}{s}{t}$ to ensure that the PLC of $i_1$ cannot
  contain $X_{i_2}$ if $i_2 > i_1$.  We can do this by iterating through space
  and time in increasing order over time (and, for fixed $t$, any order over
  space):
  \begin{align}
    \label{eq:assign_index_increasingy}
    \begin{split}
      \field{}{s}{t}, \vecS{s} \in \vecS{S}, t \in \T & \rightarrow \left( i_{(t-1) \cdot \card{ \vecS{S} } + 1}, \ldots, i_{(t-1) \cdot \card{ \vecS{S} } + \card{ \vecS{S} } )} \right) \\
      & = (t-1) \cdot \card{S} + (1, \ldots, {\card{S}}) .
    \end{split}
  \end{align}

  \begin{figure}[!b]
    \centering
    \includegraphics[width=.49\textwidth, trim = 1cm 2.5cm 0cm 0cm, clip =
    true]{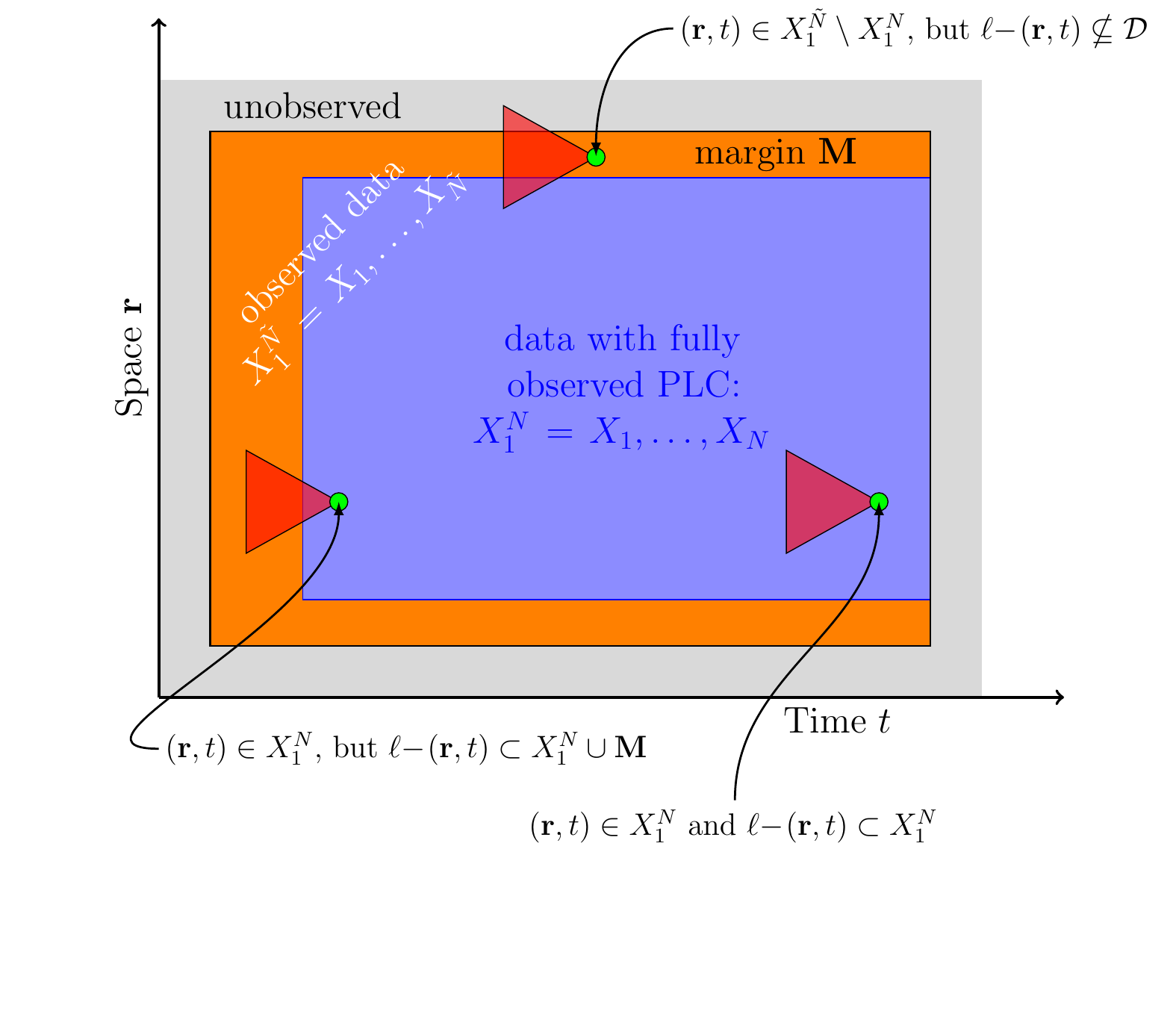}
    \caption{\label{fig:field_margin} Margin of spatio-temporal field in
      $(1+1)D$.}
  \end{figure}

  We assume that the process we observed is part of a larger field on an
  extended grid $\tilde{\vecS{S}} \times \tilde{\T}$, with $\tilde{\vecS{S}}
  \supset \vecS{S}$ and $\tilde{\T} = \lbrace - (h_p - 1), \ldots, 0, 1, \ldots
  T$, for a total of $\tilde{N} > N$ space-time points, $X_1, \ldots,
  X_{\tilde{N}}$. The margin $\vecS{M}$ are all points $\field{X}{s}{u}$,
  $\field{}{s}{u} \in \tilde{\vecS{S}} \times \tilde{\T}$ that do not have a
  fully observed past light cone.  Formally,
  \begin{equation}
    \label{eq:margin_def}
    \vecS{M} = \lbrace \field{X}{s}{u} \mid \field{\ell^{-}}{s}{u} \notin \lbrace \field{X}{r}{t}, \field{}{r}{t} \in \vecS{S} \times \T \rbrace \rbrace,
  \end{equation}
  The size of $\vecS{M}$ depends on the past horizon $h_p$ as well as the speed
  of propagation $c$, $\vecS{M} = \vecS{M}\left( h_p, c\right)$.

  In Figure \ref{fig:field_margin}, the part of the field with fully observed
  PLCs are marked in red.  Points in the margin, in gray, have PLCs extending
  into the fully unobserved region (blue).  Points in the red area have a PLC
  that lies fully in the red or partially in the gray, but never in the blue
  area.  As can be see in Fig.\ \ref{fig:field_margin} the margin at each $t$
  is a constant fraction of space, thus overall $\vecS{M}$ grows linearly with
  $T$; it does not grow with an increasing $\vecS{S}$, but stays constant.

  For simplicity, assume that $X_1, \ldots, X_N$ are from the truncated (red)
  field, such that all their PLCs are observed (they may lie in $\vecS{M}$),
  and the remaining $\tilde{N} - N$ $X_j$s lie in $\vecS{M}$ (with a PLC that
  is only partially unobserved). Furthermore let $X_1^{k} := \lbrace X_1,
  \ldots, X_k \rbrace$. Thus
  \begin{align*}
    \Prob{ \lbrace \field{X}{s}{t} \mid \field{}{s}{t} \in \tilde{\vecS{S}} \times \tilde{\T} \rbrace } & = \Prob{X_1^{\tilde{N}} } \\
    & = \Prob{X_1^{N}, \vecS{M}} \\
    & = \Prob{X_1^{N} \mid \vecS{M}} \Prob{\vecS{M}}
  \end{align*}

  The first term factorizes as
  \begin{align*}
    & \Prob{ X_1^{N} \mid \vecS{M}} \\
    & = \Prob{X_N \mid X_1^{N-1}, \vecS{M}} \Prob{X_1^{N-1} \mid \vecS{M}} \\
    &= \Prob{X_N \mid \ell^{-}_N \cup \lbrace X_1^{N-1}, \vecS{M} \rbrace \setminus \lbrace \ell^{-}_N  \rbrace} \Prob{X_1^{N-1} \mid \vecS{M}} \\
    &= \Prob{X_N \mid \ell^{-}_N} \Prob{X_1^{N-1} \mid \vecS{M}}
  \end{align*}
  where the second-to-last equality follows since by
  \eqref{eq:assign_index_increasingy}, $ \ell^{-}_N \subset \lbrace X_k \mid 1
  \leq k < N \rbrace \cup \vecS{M}$, and the last equality holds since $X_i$ is
  conditional independent of the rest given its own PLC (due to limits in
  information propagation over space-time).

  By induction,
  \begin{align}
    \Prob{X_1, \ldots, X_N \mid \vecS{M}} &= \prod_{j=0}^{N-1} \Prob{X_{N-j} \mid \ell^{-}_{N-j}} \\
    &= \prod_{i=1}^{N} \Prob{X_{i} \mid \ell^{-}_{i}}
  \end{align}

  This shows that the conditional log-likelihood maximization we use for our
  estimators is equivalent (up to a constant) to full joint maximum likelihood
  estimation.
\end{proof}

\section{Predictive States and Mixture Models}
Another way to understand predictive states is as the extremal distributions of
an optimal mixture model
\citep{Lauritzen-sufficiency-and-prediction,Lauritzen-extreme-point-models}.

To predict any variable $L^+$, we have to know its distribution $\Prob{ L^{+}
}$.  If, as often happens, that distribution is very complicated, we may try to
decompose it into a mixture of simpler ``base'' or ``extremal'' distributions,
$\Prob{L^{+} \mid \theta}$, with mixing weights $\pi(\theta)$,
\begin{equation}
  \label{eq:FLC_mixture_model_integral}
  \Prob{L^{+}} = \int{\pi(\theta) \Prob{L^{+} \mid \theta} d \theta} ~.
\end{equation}
The familiar Gaussian mixture model, for instance, makes the extremal
distributions to be Gaussians (with $\theta$ indexing both expectations and
variances), and makes the mixing weights $\pi(\theta)$ a combination of delta
functions, so that $\Prob{L^{+}}$ becomes a weighted sum of finitely-many
Gaussians.

The conditional predictive distribution of $L^{+} \mid \ell^{-}$ in
\eqref{eq:FLC_mixture_model_integral} is a weighted average over the extremal
conditional distributions $\Prob{L^{+} \mid \theta, \ell^{-}}$,
\begin{equation}
\Prob{L^{+} \mid \ell^{-}} = \int{\pi(\theta|\ell^{-}) \Prob{L^{+} \mid \theta,\ell^{-}} d \theta} ~
\end{equation}
This only makes the forecasting problem harder, unless $\Prob{ L^{+} \mid
  \theta, \ell^{-}} \pi(\theta \mid \ell^{-}) = \Prob{L^{+} \mid
  \hat{\theta}(\ell^{-})} \delta(\theta - \hat{\theta}(\ell^{-}))$, that is,
unless $\hat{\theta}(\ell^{-})$ is a predictively sufficient statistic for
$L^{+}$.  The most parsimonious mixture model is the one with the minimal
sufficient statistic, $\theta = \epsilon(\ell^{-})$.  This shows that
predictive states are the best ``parameters'' in
\eqref{eq:FLC_mixture_model_integral} for optimal forecasting.  Using them,
\begin{align}
  \label{eq:distribution_of_FLC_mixture}
  \Prob{L^{+}} &= \sum_{j=1}^{K} \Prob{\epsilon(\ell^{-}) = s_j} \Prob{L^{+} \mid \epsilon(\ell^{-}) = s_j} \\
  \label{eq:distribution_of_FLC_mixture_simplified}
  & = \sum_{j=1}^{K}{ \pi_j(\ell^{-}) \cdot \ProbState{j}{L^+} } ~,
\end{align}
where $\pi_j(\ell^{-})$ is the probability that the predictive state of
$\ell^{-}$ is $s_j$, and $\ProbState{j}{L^+} = \Prob{L^{+} \mid S = s_j}$. 
Since each light cone has a unique predictive
state,
\begin{equation}
  \label{eq:prior_probability_state}
  \pi_j(\ell^{-}) = \begin{cases}
    1, & \text { if } \epsilon(\ell^{-}) = s_j,\\ 
    0 & \text { otherwise}.
  \end{cases}
\end{equation}
Thus the predictive distribution given $\ell^{-}_i$ is just
\begin{align}
  \label{eq:distribution_of_FLC_prediction_in_mixture}
  \Prob{L^{+} \mid \ell^{-}_i} &= \sum_{j=1}^{K}{\pi_j(\ell^{-}_i) \cdot \ProbState{j}{L^{+}} }
  =  \ProbState{\epsilon(\ell^{-}_i)}{L^{+}}.
\end{align}
Now the forecasting problem simplifies to mapping $\ell^{-}_i$ to its
predictive state, $\epsilon(\ell^{-}_i) = s_j$; the appropriate
distribution-valued forecast is $p_{j}(L^+)$, and point
forecasts are derived from it as needed.\\

This mixture-model point of view highlights how prediction benefits from
grouping points by their predictive consequences, rather than by spatial
proximity (as a Gaussian mixture would do).  For us, this means clustering past
light-cone configurations according to the similarity of their predictive
distributions, not according to (say) the Euclidean geometry.  Mixed LICORS
thus learns a new geometry for the system, which is optimized for forecasting.

\end{document}